\documentclass[12pt]{amsart}

\usepackage{amsmath}
\usepackage{amssymb}
\usepackage{amsthm}
\usepackage{xcolor}
\usepackage{bm}
\usepackage{graphicx}
\usepackage{pifont}
\usepackage{mathrsfs}
\usepackage{epsfig,verbatim}
\usepackage{url}
\usepackage[pagewise]{lineno}\linenumbers

\newcommand\be{\begin{equation}}
\newcommand\ee{\end{equation}}
\newcommand\bea{\begin{eqnarray}}
\newcommand\eea{\end{eqnarray}}
\newcommand\beaa{\begin{eqnarray*}}
\newcommand\eeaa{\end{eqnarray*}}
\newcommand\beba{\begin{equation}\left\{\begin{array}{rcl}}
\newcommand\eeba{\end{array}\right.\end{equation}}
\newcommand\bebaa{\begin{equation*}\left\{\begin{array}{rcl}}
\newcommand\eebaa{\end{array}\right.\end{equation*}}

\newcommand\red{\color{red}}

\newcommand\bR{{\mathbb{R}}}

\newcommand\bN{{\mathbb{N}}}
\newcommand\bS{{\mathbb{S}}}
\newcommand\bZ{{\mathbb{Z}}}

\newcommand\cS{{\mathcal{S}}}
\newcommand\cL{{\mathcal{L}}}
\newcommand\cM{{\mathcal{M}}}
\newcommand\cK{{\mathcal{K}}}

\newcommand\cN{{\mathcal{N}}}

\newcommand\bx{{\bm{x}}}

\newcommand\by{{\bm{y}}}

\newcommand\ttheta{{\tilde{\theta}}}
\newcommand\tphi{{\tilde{\phi}}}
\newcommand\tF{{\tilde{F}}}
\newcommand\tf{{\tilde{f}}}

\newcommand{\eps}{\epsilon}

\newcommand{\dsum}{\displaystyle \sum}

\newcommand{\dlim}{\displaystyle \lim}

\newcommand{\lng}{\left\langle  \right.}
\newcommand{\rng}{\left.  \right\rangle}
\newtheorem{theorem}{Theorem}[section]

\newtheorem{corollary}[theorem]{Corollary}
\newtheorem{example}[theorem]{Example}

\newtheorem{remark}[theorem]{Remark}
\newtheorem{proposition}[theorem]{Proposition}
\newtheorem{assumption}[theorem]{Assumption}

\numberwithin{equation}{section}
\def\theequation{\arabic{section}.\arabic{equation}}

\begin{document}
\nolinenumbers

\title[Spot solutions on oblate spheroids]{Spot solutions to a neural field equation \\ on oblate spheroids}

\author[H. Ishii]{Hiroshi Ishii}
\address{Research Center of Mathematics for Social Creativity, Research Institute for Electronic Science, Hokkaido University, Hokkaido, 060-0812, Japan}
\email{hiroshi.ishii@es.hokudai.ac.jp}

\author[R. Watanabe]{Riku Watanabe}
\address{Department of Mathematics, Faculty of Science, Hokkaido University, Hokkaido, 060-0810, Japan}
\email{watanabe.riku.y9@elms.hokudai.ac.jp}

\thanks{Date: \today. Corresponding author: Hiroshi Ishii}

\thanks{{\em Keywords. Neural field equations, Spot pattern, Spheroid, Integro-differential equations}}
\thanks{{\em 2020 MSC} Primary 35R09, Secondary 35P20, 45L05}

\begin{abstract}
Understanding the dynamics of excitation patterns in neural fields is an important topic in neuroscience. 
Neural field equations are mathematical models that describe the excitation dynamics of interacting neurons to perform the theoretical analysis.
Although many analyses of neural field equations focus on the effect of neuronal interactions on the flat surface, the geometric constraint of the dynamics is also an attractive topic when modeling organs such as the brain. 
This paper reports pattern dynamics in a neural field equation defined on spheroids as model curved surfaces.
We treat spot solutions as localized patterns and discuss how the geometric properties of the curved surface change their properties. 
To analyze spot patterns on spheroids with small flattening, 
we first construct exact stationary spot solutions on the spherical surface and reveal their stability. 
We then extend the analysis to show the existence and stability of stationary spot solutions in the spheroidal case.
One of our theoretical results is the derivation of a stability criterion for stationary spot solutions localized at poles on oblate spheroids. 
The criterion determines whether a spot solution remains at a pole or moves away.
Finally, we conduct numerical simulations to discuss the dynamics of spot solutions with the insight of our theoretical predictions. 
Our results show that the dynamics of spot solutions depend on the curved surface and the coordination of neural interactions.
\end{abstract}

\maketitle
\setcounter{tocdepth}{3}

\section{Introduction}\label{sec:int}


Neural field theory provides a macroscopic modeling approach to describe the spatial and temporal evolution of neural activity in terms of population-level variables such as membrane potential or firing rate. 
A variety of integro-differential equations, commonly known as neural field equations, have been introduced to capture the dynamics of excitation patterns in spatially distributed neuronal populations (e.g., \cite{Bress, CPWT, Murray, Potthast}).
Among them, the Amari model is a representative example, characterized by neuronal interactions mediated through spatial connection kernels \cite{Amari}.
The spatially homogeneous Amari model is given by the following on the domain $\Omega$:
\beaa
u_t (t,\bx) = \int_{\Omega} K(|\bx-\by|) P(u(t,\by)) d\by - u(t,\bx) \quad (t>0,\ \bx\in\Omega),
\eeaa
where $u(t,\bx)$ is the average activity of the neuronal population at time $t>0$ and position $\bx\in\Omega$, 
the connection kernel $K$ represents the strength of neural interactions depending on the distance,
and $P(u)$ is the mean firing rate.
Typical examples of $P(u)$ are a sigmoid function or a step function
\beaa
P(u) = H(u-u_T)
\eeaa
with the threshold $u_T\in\bR$ of the neural activity and Heaviside step function $H(u)$.
This neural field equation was not directly derived from a microscopic model of neuronal spiking processes, but was derived phenomenologically as a neural mean field \cite{Amari}.
The Amari model and related models have been reported to generate a wide variety of spatio-temporal dynamics. 
Theoretical studies to understand the dynamics of patterns have been conducted through the construction of stationary and traveling wave solutions, and the analysis of the interface dynamics \cite{Amari, Bress, Bress2, CSA, GC1, GC2, GZ, YZ}.

It is more natural to consider neural field equations on curved surfaces when considering organs such as the brain.
In particular, it is attractive to study how the patterning changes under the geometric constraints of curved surfaces.
The effects of geometry have also been discussed in the context of neuroscience (e.g. \cite{MCC, PAO}).
In the Amari model and related models, 
the relationship between geometrical properties and patterning has been considered in numerical simulations \cite{MCC}.
Furthermore, in the case of the unit spherical surface, 
a theoretical analysis of the linearized stability of constant states has also been carried out by \cite{VNFC}.
While theoretical analysis of pattern formation in neural field equations on curved surfaces has been advanced, to the best of our knowledge, there are few analytical studies for the case of curved surfaces with nonconstant curvature.

As related problems in reaction-diffusion models, 
pattern formation has been analyzed on spherical surfaces \cite{RRW, TW}, torus \cite{SW, TT}, and deformed surfaces \cite{NI1, NI2}.
Notably, it has been reported that a stationary pattern on a symmetric surface can become a dynamic pattern due to the asymmetry of the surface \cite{NI1, NI2}.
As reported in these studies, the properties of curved surfaces are known to have a significant impact on pattern behavior.

In this paper, we consider pattern formation in the Amari model defined on curved surfaces with the aim of studying how geometry of curved surfaces affects the solution behavior.
In particular, we focus on spot solutions as localized patterns and examine how they behave on spheroids with small flattening as model surfaces.
We are interested in how the properties of the spot solution formed on the sphere are changed by the perturbation that transforms it into a spheroid.
In the Amari model defined on general curved surfaces, 
the properties of curved surfaces appear in the integral term that describes the interaction between neurons.
To study the behavior of solutions, 
it is necessary to investigate the relationship between the global information of the curved surface and the connection kernel, 
which is a difficult part of the analysis.

Our analysis is based on a perturbation expansion. 
First, we construct a spot solution on the spherical surface and discuss its stability. 
Then, we construct a spot solution localized at the pole of oblate spheroids with small flattening and derive its stability criterion.
The derived stability criterion is determined from a quantity based on information about the curved surface and the connection kernel.
Based on the predictions of the theoretical analysis, numerical simulations are performed and the behavior of the spot solution is discussed.

The paper is organized as follows:
the problem setting is explained in Section \ref{sec:pro}.
We present results on the existence and stability of stationary spot solutions on the unit spherical surface $\bS^2$ in Section \ref{sec:s2}.
In Section \ref{sec:cri} we derive a stability criterion for stationary spot solutions localized at the north pole on spheroids with small flattening.
Section \ref{sec:num} shows the usefulness of the stability criterion, 
as well as the behavior of the spot solution by numerical simulation.
We summarize the results and some of the remaining issues in Section \ref{sec:dis}.

\bigskip
\section{Problem setting}\label{sec:pro}

Let us explain the setting of our problem.
We first formulate Amari model on spheroids.
The unit spherical surface $\bS^2$ is parameterized as follows:
\beaa
\bS^2:= \{ \bx(\theta,\phi)= (\sin\theta\cos\phi, \sin\theta\sin\phi, \cos\theta)\ 
\mid\ (\theta,\phi)\in [0,\pi]\times [0,2\pi)\ \}.
\eeaa
For a sufficiently small constant $\eps>0$,
we treat an oblate spheroid with small flattening represented by
\beaa
x^2 + y^2 + \left( \dfrac{z}{1-\eps} \right)^2 =1,\quad (x,y,z)\in\bR^3.
\eeaa
Let us define 
\be\label{surface}
S := \left\{ \bx_{\eps}(\theta,\phi) =  (1-\eps p(\eps,\bx(\theta,\phi))) \bx(\theta,\phi)\ \mid\ \bx(\theta,\phi)\in\bS^2,\  (\theta,\phi)\in [0,\pi]\times [0,2\pi) \right\}
\ee
as the oblate spheroid with the same parameters as $\bS^2$, where
\be\label{per}
p(\eps,\bx(\theta,\phi)) := \dfrac{1}{\eps} \left\{ 1 -  \dfrac{1-\eps}{\sqrt{\cos^2\theta + (1-\eps)^2 \sin^2\theta}} \right\}.
\ee
It is easy to see that
\beaa
\sup_{\theta\in[0,\pi]}|p(\eps,\bx(\theta,\phi)) - \cos^2\theta | = O(\eps)
\eeaa
as $\eps\to +0$.
We note that $p(\eps,\bx(\theta,\phi))$ is a function of $\eps$ and $\theta$, independent of $\phi$.
However, some calculations can proceed with the more general $p(\eps,\bx)$, and thus we will proceed without using the concrete form as much as possible.

We treat the following neural field equation on $S$:
\be\label{eq:main}
\dfrac{\partial u}{\partial t}(t,\bx_{\eps}) = \int_{S} K(d(\bx_{\eps},\by_{\eps}))H(u(t,\by_{\eps})-u_T) ds(\by_{\eps}) - u(t,\bx_{\eps})\quad (t>0,\ \bx_{\eps}\in S),
\ee
where $d(\bx_{\eps},\by_{\eps})$ is the geodesic distance on $S$, $u_T\in\bR$, and $K:\bR\to\bR$ is a sufficiently smooth function.
Furthermore, $ds(\by_{\eps})$ is the element of area on $S$ defined as
\beaa
ds(\by_{\eps}(\ttheta,\tphi)) := \left| \dfrac{\partial \by_{\eps}}{\partial \ttheta}(\ttheta,\tphi) \times \dfrac{\partial \by_{\eps}}{\partial \tphi}(\ttheta,\tphi) \right| d\ttheta d\tphi.
\eeaa
Note that for the geodesic distance on $S$, an approximate formula is given by \eqref{exp:dist}.

For the analysis of spot solutions and numerical simulations, 
we treat
\be\label{kernel}
K(x)= \dsum^{3}_{n=0} c_n \cos n x 
\ee
with $c_n\in\bR\ (n=0,1,2,3)$ as an example of $K$, which is easy to handle. 
In specific calculation examples, we will mainly deal with connection kernels called Mexican hat-type kernels, which are positive near the origin and negative away from the origin, as shown in Fig. \ref{fig:kernel}.
This type of connection kernels has excitatory connections between neighboring neurons, but inhibitory connections to distant neurons.
Such connection kernels are often used as simplified representations of neural interactions to explain pattern formation in neural fields \cite{Amari, CSB, GC1, Murray}.

\begin{figure}[bt] 
\begin{center}
	\includegraphics[width=14cm]{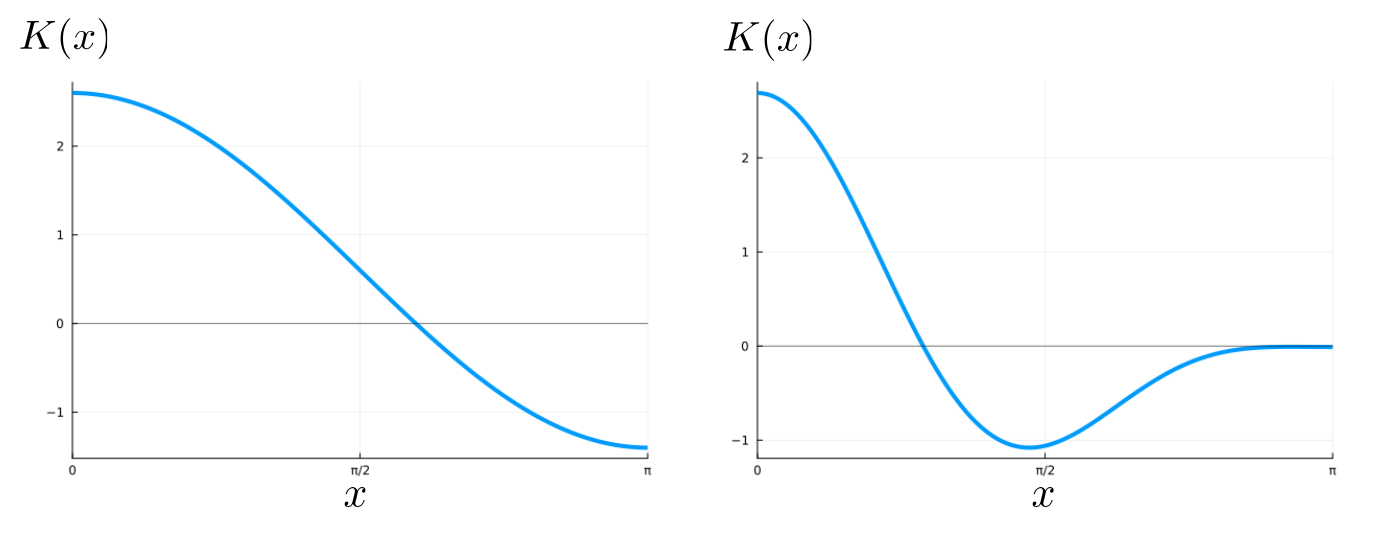}
\end{center}
\caption{\small{
The graph of \eqref{kernel} with Mexican hat-type shapes. 
Left: $c_0=0.6$, $c_1=2.0$, $c_2=c_3=0.0$.
Right: $c_0=0.14$, $c_1=0.9$, $c_2=1.2$, $c_3=0.45$.
}}
\label{fig:kernel}
\end{figure}

\bigskip
\section{Stationary spot solutions on the unit spherical surface}\label{sec:s2}

To analyze spot solutions on spheroids with small flattening,
we first consider the properties of stationary spot solutions
on $\bS^2$.
We use a perturbation method with flattening $\eps>0$ as a small parameter to construct spot solutions on the spheroid in Section \ref{sec:cri}.
For this reason, it is essential to analyze the existence and linear stability of the spot solution when $\eps=0$.
Furthermore, since the implicit function theorem is used in the perturbation problem, it is necessary as a condition that the linearized operator around the spot solution is invertible in some sense.

For the one-dimensional space,
various analytical results on existence, stability, and weak interactions for pulse solutions have been reported
(e.g. \cite{Amari, Bress2, GC1, GC2}).
In the case of the plane, the existence and stability of spot solutions have been discussed in \cite{CSB, WR}.

In this section, we extend their discussions to the case of $\bS^2$. 
We construct stationary spot solutions and discuss the characterization of their linear stability.
We also mention the Fredholm property of a linearized operator around a stationary spot solution, 
which we will use for the implicit function theorem in Section \ref{sec:cri}.

\subsection{Construction of stationary spot solutions}

Let us construct stationary spot solutions for the case $\eps=0$.
In this case, the equation we consider is given as
\be\label{eq:sphe}
u_t(t,\bx) = \int_{\bS^2} K(d_0(\bx,\by))H(u(t,\by)-u_T) ds_0(\by) - u(t,\bx)\quad (t>0,\ \bx\in \bS^2).
\ee
Here, $d_0(\bx,\by)$ is the great-circular distance defined as
\bea
d_0(\bx(\theta,\phi),\by(\ttheta,\tphi)) &:=& \cos^{-1} (\bx(\theta,\phi)\cdot \by(\ttheta,\tphi))\notag  \\
&=& \cos^{-1} (\sin\theta\sin\ttheta\cos(\phi-\tphi)+\cos\theta\cos\ttheta). \label{dist0}
\eea
Furthermore, $ds_0(\by)$ is the element of area on $\bS^2$ defined as
\beaa
ds_0(\by(\ttheta,\tphi)) := \sin\ttheta d\ttheta d\tphi.
\eeaa
Since $\bS^2$ has rotational symmetry,
we only treat stationary spot solutions located at the north pole $N:=(0,0,1)\in \bS^2$.
Note that $d_0(\bx(\theta,\phi),N)=\theta$.

Consider the existence of a stationary spot solution $u(t,\bx) = U({\red d_0}(\bx,N))$ satisfying
\be\label{eq:spot}
U(\theta;\theta_T) 
\begin{cases}
>u_T &(\theta < \theta_T) \\ 
=u_T &(\theta = \theta_T) \\
<u_T &(\theta > \theta_T)
\end{cases}
\ee
for some $\theta_T\in (0,\pi)$.
For $\bx=\bx(\theta,\phi)\in\bS^2$ and $\by=\by(\ttheta,\tphi)\in\bS^2$,
this spot solution is represented by
\bea
U(\theta;\theta_T) &=& \int^{2\pi}_{0}\int^{\theta_T}_{0} K(d_0(\bx(\theta,\phi),\by(\ttheta,\tphi))) \sin\ttheta d\ttheta d\tphi \notag \\
&=& \int^{\theta_T}_{0} \left( \int^{2\pi}_{0} K(\cos^{-1} (\sin\theta\sin\ttheta\cos\tphi+\cos\theta\cos\ttheta)) 
d\tphi \right) \sin \ttheta d\ttheta{\red .} \label{sol}
\eea
Thus, if we find $\theta_T$ satisfying \eqref{eq:spot},
then $U(\theta;\theta_T)$ is a spot solution to \eqref{eq:sphe}.
For general connection kernels, it is not easy to ensure the existence of stationary spot solutions. However, in the following case, the spot solution can be constructed:

\medskip
\begin{example}\label{ex:exi}
    We consider the case where $K(x)$ is given by \eqref{kernel}.
    Then, we obtain
    \beaa
    U(\theta;\theta_T)
    &=& 2\pi(c_0-c_2)(1-\cos\theta_T)
    + \pi (c_1-3c_3)\sin^2\theta_T\cos\theta \\
    &&\quad 
    + 2\pi c_2\Big\{\sin^2\theta_T\cos\theta_T\cos^2\theta 
    + \frac{\cos^3\theta_T - 3\cos\theta_T + 2}{3}\Big\} \\
    &&\quad +\pi c_3\Big\{ 2(1-\cos^4\theta_T)\cos^3\theta
    + 3\sin^4\theta_T\cos\theta\sin^2\theta\Big\}
    \eeaa
    from the direct computation.

    To obtain the existence result,
    we treat the case $c_1>0,c_2=c_3=0$.
    In this case, $U(\theta;\theta_T)$ is a strictly decreasing function with respect to $\theta$.
    Hence, a spot solution satisfying \eqref{eq:spot} is constructed if there is $\theta_T\in(0,\pi)$ such that 
    \be\label{condi:exa}
    U(\theta_T;\theta_T)= 2\pi c_0(1-\cos\theta_T)
    + \pi c_1 \sin^2\theta_T\cos\theta_{T} = u_T
    \ee
    for $u_T\in\bR$.
    We see that the solution $\theta_T\in(0,\pi)$ to \eqref{condi:exa} exists if and only if any of the following conditions is satisfied.
    \begin{itemize}
        \item[(i)] $c_1\leq 2c_0,\ 0<u_T<4\pi c_0$;
        \item[(iia)] $8c_0>c_1> 2c_0$, $c_0+c_1>0$, $0<u_T<4\pi c_0$;
        \item[(iib)] $c_1\ge 8c_0$, $c_0+c_1>0$,
        \beaa
         \pi\left(2c_0-\dfrac{2}{3}\sqrt{\dfrac{c_1-2c_0}{3c_1}}(c_1-2c_0)\right)\leq u_T\leq  \pi\left(2c_0+\dfrac{2}{3}\sqrt{\dfrac{c_1-2c_0}{3c_1}}(c_1-2c_0)\right);
        \eeaa
        \item[(iii)] $c_0+c_1 \leq 0,\quad 4\pi c_0<u_T<0$.
    \end{itemize}
    See Appendix \ref{app:ex} for the derivation of the conditions.
    The existence conditions for some $u_T\in\bR$ are shown in Fig. \ref{fig:exi}.

    In the case that $(c_2,c_3)\neq(0,0)$, we can confirm the existence numerically. Numerical results are presented in Section \ref{sec:num}.
\end{example}

\begin{figure}[bt] 
\begin{center}
	\includegraphics[width=14cm]{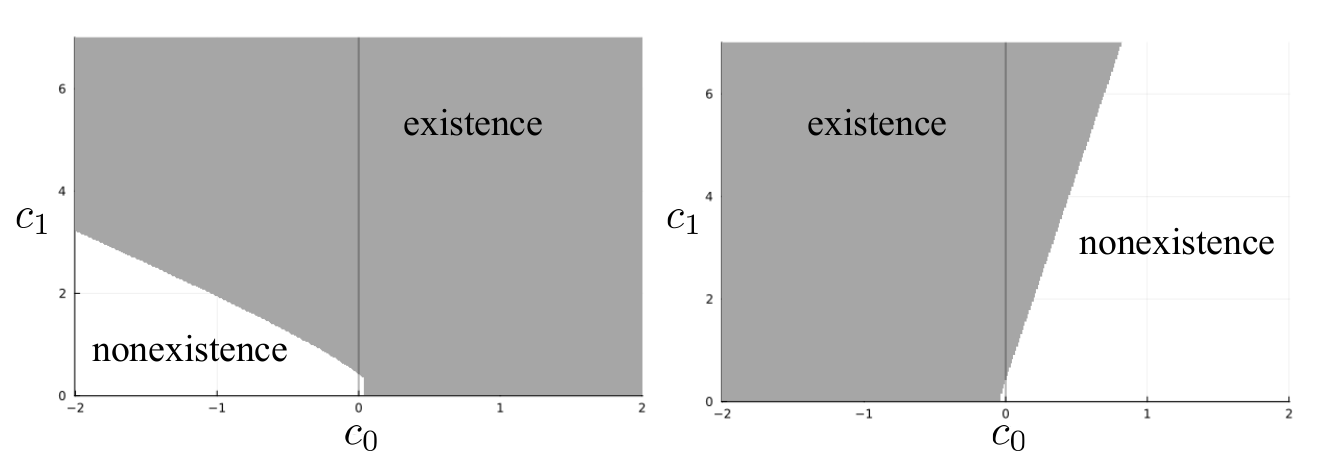}
\end{center}
\caption{\small{
Parameter areas for the existence of stationary spot solutions \eqref{eq:spot} when the kernel is \eqref{kernel} with $c_1>0$ and $c_2=c_3=0$.
Left: $u_T=0.5$. Right: $u_T=-0.5$.
}}
\label{fig:exi}
\end{figure}

\medskip
\subsection{Stability of stationary spot solutions}

Let us discuss the linearized stability of the stationary spot solutions constructed in the previous subsection.
To define the linearized operator around the spot solution,
we introduce the following assumption: 
\begin{assumption}\label{ass:exi}
There exists a stationary spot solution $U(\theta)=U(\theta;\theta_T)$ to \eqref{eq:sphe} satisfying \eqref{eq:spot} and
$U'(\theta_T)< 0$.
\end{assumption}
Note that a stationary spot solution satisfying \eqref{eq:spot} implies that $U'(\theta_T)\le 0$ is satisfied.
As we show below, our linearized stability analysis is not applicable for the case $U'(\theta_T)=0$,
and thus we introduce Assumption \ref{ass:exi}.
Stationary spot solutions that satisfy this assumption were constructed in Example \ref{ex:exi}.

\medskip
Let Assumption \ref{ass:exi} be satisfied.
Then, we define the linearized operator $L_0=L_0(\theta_T):C(\bS^2)\to C(\bS^2)$ around the spot solution as 
\beaa
L_0 v(\bx) := \int_{\bS^2} K(d_0(\bx,\by))\delta(U(d_0(\by,N))-u_T) v(\by) ds_0(\by) - v(\bx).
\eeaa
Here, $\delta(x)$ is Dirac's delta function.
Define the integral operator
$\cK=\cK(\theta_T):C(\bS^2)\to C(\bS^2)$ as
\beaa
\cK v(\bx) := \int_{\bS^2}K(d_0(\bx,\by))\delta(U(d_0(\by,N))-u_T) v(\by) ds_0(\by).
\eeaa
For $\bx=\bx(\theta,\phi),\ \by=\by(\ttheta,\tphi)\in\bS^2$,
we have
\beaa
\cK v(\bx(\theta,{\red \phi})) 
&=& \int^{\pi}_{0}\int^{2\pi}_{0} K(d_0(\bx(\theta,\phi),\by(\ttheta,\tphi)))\delta(U(\ttheta)-u_T) v(\by(\ttheta,\tphi)) \sin\ttheta d\tphi d\ttheta\\
&=& \dfrac{\sin\theta_T}{|U'(\theta_T)|} \int^{2\pi}_{0} K(d_0(\bx(\theta,\phi),\by(\theta_T,\tphi))) v(\by(\theta_T,\tphi)) d\tphi.
\eeaa
Note that $\cK$ is well defined from Assumption \ref{ass:exi}.

Our goal of this subsection is to compute the spectrum of $L_0$ in order to characterize the linearized stability of the spot solution.
Since $\cK = L_0 + I$ holds,
it suffices to analyze the spectrum of $\cK$ 
for the stability problem.
Denote the spectrum of $\cK:C(\bS^2)\to C(\bS^2)$ by $\sigma(\cK)$.
We note that $\lambda\in\sigma(\cK)$ with $\lambda\neq 0$ is an eigenvalue, since $\cK$ is a compact operator \cite[Theorem 3.9]{Kress}.
Moreover, $0\in \sigma(\cK)$ holds.

To analyze the eigenvalues of $\cK$,
we introduce the properties of spherical harmonics.
Spherical harmonics $Y^m_l=Y^m_l(\bx(\theta,\phi))\ (l=0,1,2\ldots,\ |m|\le l)$ are defined as follows:
\beaa
Y^m_l({\red \bx}(\theta,\phi)):=(-1)^m\frac{C^{|m|}_l}{\sqrt{2\pi}} e^{im\phi}P^{|m|}_l(\cos\theta),
\eeaa
where $P^m_l(x)\ (l=0,1,2\ldots,\ 0\le m\le l)$ is the associated Legendre polynomial
\beaa
P^m_l(x)&:=& \dfrac{1}{2^l l!}(1-x^2)^{m/2}
\dfrac{d^{m+l}}{dx^{m+l}}\left[(x^2-1)^l\right],
\eeaa
and $C^{m}_l\ (l=0,1,2\ldots,\ 0\le m\le l)$ is a normalization constant
\beaa
C^m_l &:=\sqrt{\left(l+\dfrac{1}{2}\right)\dfrac{(l-m)!}{(l+m)!}} \quad(m\ge 0).
\eeaa
Then, spherical harmonics have the following properties.
\begin{proposition}[\cite{AT}] \label{prop:sphe}
$\{Y^m_l(\bx)\mid l=0,1,2\ldots,\ |m|\le l\}$ is an orthonormal system in $L^2(\bS^2)$. Moreover, we have the following:
\begin{itemize}
    \item[(i)] $\mathrm{span}\{Y^m_l(\bx)\mid l=0,1,2\ldots,\ |m|\le l\}$ is dense in $C(\bS^2)$;
    \item[(ii)] Legendre polynomials $P_l(x):= P^{0}_l(x)\ (l\ge 0)$ are represented by 
    \beaa
    P_l(\bx\cdot\by)=\frac{4\pi}{2l+1}\sum^l_{m=-l}Y^m_l(\bx)\overline{Y^m_l(\by)}, \quad (\bx, \by\in \bS^2).
    \eeaa
\end{itemize}
\end{proposition}


By applying the polynomial approximation theory and Proposition \ref{prop:sphe}, 
there are real constants $k_n\ (n=0,1,\ldots)$ such that
for any $\bx,\by\in\bS^2$, we obtain
\bea
K(d_0(\bx,\by)) &=& K(\cos^{-1}(\bx\cdot\by)) 
= \dsum^{\infty}_{n=0} k_n P_n (\bx\cdot\by) \notag \\ 
&=&  4\pi \dsum^{\infty}_{n=0}\frac{k_n}{2n+1}\sum^n_{m=-n}Y^m_n(\bx)\overline{Y^m_n(\by)} \label{kernel-exp}
\eea
in the sense of the uniform convergence.
Then, we find
\beaa
K(d_0(\bx(\theta,\phi),\by(\ttheta,\tphi)) &=&  2 \dsum^{\infty}_{n=0}\frac{k_n}{2n+1}\sum^n_{m=-n}(C^{|m|}_{n})^2 P^{|m|}_{n}(\cos\theta) P^{|m|}_{n}(\cos\ttheta)  e^{i m(\phi-\tphi)} \notag \\
&=&  2 \dsum_{m\in\bZ} \left( \sum^{N}_{n=|m|} \frac{k_n}{2n+1}(C^{|m|}_{n})^2 P^{|m|}_{n}(\cos\theta) P^{|m|}_{n}(\cos\ttheta) \right) e^{im(\phi-\tphi)} \notag \\
&=&  2 \left( \sum^{\infty}_{n=0} \frac{k_n}{2n+1}(C^{0}_{n})^2 P_{n}(\cos\theta) P_{n}(\cos\ttheta) \right) \notag \\
&& \quad + 4 \dsum^{\infty}_{m=1} \left( \sum^{\infty}_{n=m} \frac{k_n}{2n+1}(C^{m}_{n})^2 P^{m}_{n}(\cos\theta) P^{m}_{n}(\cos\ttheta) \right)  \cos m(\phi-\tphi). 
\eeaa
Here, by setting
\beaa
K_0(\theta,\ttheta) &:=& 2 \sum^{\infty}_{n=0} \frac{k_n}{2n+1}(C^{0}_{n})^2 P_{n}(\cos\theta) P_{n}(\cos\ttheta), \\ 
K_m(\theta,\ttheta) &:=& 4 \sum^{\infty}_{n=m} \frac{k_n}{2n+1}(C^{m}_{n})^2 P^{m}_{n}(\cos\theta) P^{m}_{n}(\cos\ttheta)\quad (m\ge 1),
\eeaa
we have
\be\label{exp}
K(d_0(\bx(\theta,\phi),\by(\ttheta,\tphi))= \dsum^{\infty}_{m=0} K_m(\theta,\ttheta)\cos m(\phi-\tphi).
\ee

With the above preparations, let us analyze the eigenvalues of $\cK$.
First, since the spot solution can be constructed with any position on the spherical surface as its center, there are two degrees of freedom associated with its location.
Thus, $L_0$ has the eigenvalue $0$ corresponding to degrees of freedom. 
This implies that $\cK$ has the following eigenvalues and eigenfunctions corresponding to the translative invariance.
\begin{proposition}\label{prop:tra}
Functions
\be\label{trans}
v^1(\bx(\theta,\phi))= U'(\theta)\cos\phi,\quad v^2(\bx(\theta,\phi)) = U'(\theta)\sin\phi
\ee
are eigenfunctions corresponding to the eigenvalue $\lambda=1$ of $\cK$.
Moreover, we find
\be\label{sol:deri}
U'(\theta) = -\pi\sin\theta_T K_1(\theta,\theta_T).
\ee
\end{proposition}
\begin{proof}
    Since the spot solution can be translated due to the symmetry of the sphere, 
    $U(d_0(\bx,C))$ for $C\in\bS^2$ is also a stationary spot solution to \eqref{eq:sphe}.
    That is, 
    \beaa
    U(d_0(\bx,C))=\int_{\bS^2} K(d_0(\bx,\by))H(U(d_0(\by,C)-u_T)) ds_0(\by)
    \eeaa
    holds for any $C=C(\theta_c,\phi_c)$.
    By calculating the derivative of $U(d_0(\bx,C(\theta_c,\phi_c)))$ with respect to $\theta_c=0$, we obtain
    \beaa
    \dlim_{\theta_c\to +0}\dfrac{\partial}{\partial \theta_c} [U(d_0(\bx(\theta,\phi),C(\theta_c,\phi_c))] =-\cos(\phi-\phi_c)U'(\theta).
    \eeaa
    Thus, $v^1(\bx)$ and $v^2(\bx)$ are eigenfunctions corresponding to the eigenvalue $\lambda=1$ of $\cK$.
    Also, we find 
    \beaa
    {\red v^1}(\bx(\theta,\phi))= \cK v^1(\bx(\theta,\phi)) 
    =-\pi\sin\theta_T K_1(\theta,\theta_T)\cos\phi.
    \eeaa
    Hence \eqref{sol:deri} is obtained.
\end{proof}
\begin{remark}
    Under Assumption \ref{ass:exi}, $K_{1}(\theta_T,\theta_T)$ is positive from \eqref{sol:deri} since $\theta_T\in(0,\pi)$ holds.
\end{remark}

\medskip
Next, using the spherical harmonic expansion, all elements of $\sigma(\cK)$ are computed as follows.
\begin{proposition}\label{prop:sta}
$\sigma(\cK)$ is represented by
\beaa
\sigma(\cK) = \{0,1\} 
\cup \left\{ \dfrac{2 K_0(\theta_T,\theta_T)}{K_1(\theta_T,\theta_T)} \right\}
\cup \left\{ \dfrac{K_m(\theta_T,\theta_T)}{K_1(\theta_T,\theta_T)}\ \Big|\ m\ge 2 \right\}.
\eeaa
\end{proposition}
\begin{proof}
Since $\cK$ is a compact operator on $C(\bS^2)$, it follows that $0\in\sigma(\cK)$, and the rest of the spectrum is an eigenvalue.
Thus, it is enough to discuss the eigenvalue of $\sigma(\cK)$.
For $\alpha\in\bZ$, we define
\beaa
E_{\alpha} := \overline{\mathrm{span}\{ Y^{\alpha}_{\beta}(\bx) \mid \beta \ge |\alpha|\}}\ \mathrm{in}\ C(\bS^2).
\eeaa
For every $Y^{\alpha}_{\beta}(\bx)\ (\beta=0,1,2\ldots,\ |\alpha|\le \beta)$,
\beaa
\cK Y^{\alpha}_{\beta}(\bx) 
&=& \dfrac{4\pi \sin\theta_T}{|U'(\theta_T)|}  \dsum^{\infty}_{n=0} \frac{k_n}{2n+1}\sum^n_{m=-n}  Y^m_n(\bx) \int^{2\pi}_{0} Y^{\alpha}_{\beta}(\by(\theta_T,\tphi))\overline{Y^m_n(\by(\theta_T,\tphi))} d\tphi 
\eeaa
holds from \eqref{kernel-exp}.
From the definition of spherical harmonics,
we obtain
\beaa
\int^{2\pi}_{0} Y^{\alpha}_{\beta}(\by(\theta_T,\tphi))\overline{Y^m_n(\by(\theta_T,\tphi))} d\tphi
&=& \dfrac{(-1)^{\alpha-m}}{2\pi} C^{|\alpha|}_{\beta} C^{|m|}_{n} P^{|\alpha|}_{\beta}(\cos\theta_T)  P^{|m|}_{n}(\cos\theta_T)  \int^{2\pi}_{0} e^{i(\alpha-m)\tphi} d\tphi \\
&=& 
\begin{cases}
     C^{|\alpha|}_{\beta} C^{|\alpha|}_{n} P^{|\alpha|}_{\beta}(\cos\theta_T)  P^{|\alpha|}_{n}(\cos\theta_T) & (\alpha=m) \\
    0 &(\alpha\neq m).
\end{cases}
\eeaa
This means $\cK : E_{\alpha}\to E_{\alpha}$.
Denote the restricted operator of $\cK$ on $E_{\alpha}$ by $\cK_{\alpha}$.  

We know from Proposition \ref{prop:sphe} that 
\beaa
E_{\alpha_{1}} \perp E_{\alpha_{2}}\quad \mathrm{in}\ L^2(\bS^2)
\eeaa
holds for distinct integers $\alpha_1, \alpha_2$.
Thus, if $v\in C(\bS^2)$ is an eigenfunction corresponding to the {\red eigenvalue} $\lambda\in \sigma(\cK)$,
then there exists $\alpha\in\bZ$ such that $\lambda$ is an eigenvalue of $\cK_{\alpha}$.
Therefore, it suffices to analyze the eigenvalue of the operator $\cK : E_{\alpha}\to E_{\alpha}$ for all $\alpha\in\bZ$.

Fix $\alpha\in\bZ$.
Let $\lambda$ be an eigenvalue of $\cK_{\alpha}$ and
$v_{\alpha}\in E_{\alpha}$ be the eigenfunction.
Then there are $v_{\alpha,\beta}\in\bR\ (\beta\ge |\alpha|)$ such that
$v_{\alpha}(\bx(\theta,\phi))= \dsum^{\infty}_{\beta=|\alpha|} v_{\alpha,\beta} Y^{\alpha}_{\beta}(\bx)$.
Moreover, setting
\beaa
g_{\alpha}(\theta) := \dsum^{\infty}_{\beta=|\alpha|} (-1)^{\alpha}  \dfrac{C^{|\alpha|}_{\beta} v_{\alpha,\beta}}{\sqrt{2\pi}} P^{|\alpha|}_{\beta}(\cos\theta)
\eeaa
yields $g_{\alpha}\in C[0,\pi]$ and
\beaa
v_{\alpha}(\bx(\theta,{\red \phi})) = g_{\alpha}(\theta) e^{i\alpha\phi}.
\eeaa
Thus, we find
\beaa
\cK v_{\alpha}(\bx(\theta,\phi)) &=&  \dfrac{\sin\theta_T}{|U'(\theta_T)|} g_{\alpha}(\theta_T)\dsum^{\infty}_{m=0} K_m(\theta,\theta_T) \int^{2\pi}_{0} \cos m(\phi-\tphi) e^{i\alpha\tphi} d\tphi.
\eeaa
Note that
\beaa
\int^{2\pi}_{0} \cos m(\phi-\tphi) e^{i\alpha\tphi} d\tphi = 
\begin{cases}
    2\pi \delta_{0,m} e^{im\phi} & (\alpha =0), \\
    \pi \delta_{|\alpha|,m} e^{im\phi} &(\alpha\neq 0),
\end{cases}
\eeaa
where $\delta_{\alpha,m}$ is Kronecker delta.

Here, we have 
$|U'(\theta_T)| = \pi \sin\theta_{T} K_1(\theta_T,\theta_T)$
from \eqref{sol:deri}.
When $\alpha=0$, it follows that
\beaa
\dfrac{2K_0(\theta,\theta_T)}{K_1(\theta_T,\theta_T)} g_{0}(\theta_T) = \lambda g_{0}(\theta)
\eeaa
and conclude that
\beaa
\dfrac{2K_0(\theta_T,\theta_T)}{K_1(\theta_T,\theta_T)} \in\sigma(\cK).
\eeaa
In the case that $\alpha\neq 0$,
we obtain
\beaa
\dfrac{K_{|\alpha|}(\theta_T,\theta_T)}{K_1(\theta_T,\theta_T)} \in \sigma(\cK)
\eeaa
in the same way.

Therefore, the desired assertion is shown, since all elements in $\sigma(\cK)$ has been computed from the above discussion.
\end{proof}

\begin{remark}
    In $\sigma(\cK)$ shown in Proposition \ref{prop:sta}, the case $m=1$ corresponds to the eigenvalue shown in Proposition \ref{prop:tra}.
\end{remark}

\medskip
Let us discuss the linearized stability of the spot solution. 
In this case, the spectrum of $L_0$ must have a negative real part, except for the eigenvalue corresponding to the translation degrees of freedom.
It is equivalent to the real part of the eigenvalues of $\cK$ being less than $1$.
Thus, using the fact that $K_{m}(\theta_T,\theta_T)\ (m\ge 0)$ are real constants, 
we obtain the linearized stability condition for a spot solution to \eqref{eq:sphe} under Assumption \ref{ass:exi} as
\be\label{condi:sta}
\dfrac{2K_0(\theta_T,\theta_T)}{K_1(\theta_T,\theta_T)} <1,\quad
\dfrac{K_m(\theta_T,\theta_T)}{K_1(\theta_T,\theta_T)} <1\quad (m\ge 2)
\ee
from \eqref{sol:deri}.
Since the spectrum has been calculated by concrete computations, the following fact is easily derived.

\begin{corollary}\label{cor:ker}
    Let Assumption \ref{ass:exi} be satisfied.
    If the spot solution \eqref{eq:spot} satisfies \eqref{condi:sta},
    then 
    \beaa
    \mathrm{Ker} L_0 = \mathrm{span}\{v^1,v^2\}
    \eeaa
    holds, where $v^{j}(\bx)\ (j=1,2)$ are defined by \eqref{trans}. 
\end{corollary}

\medskip
Let us check the linearized stability of spot solutions constructed in Example \ref{ex:exi}.

\medskip
\begin{example}\label{ex:sta}
    We consider again the case that $K(x)$ is given by \eqref{kernel}.
    Then, we find
    \beaa
    K_0(\theta,\ttheta)
    &=& c_0-c_2 + (c_1-3c_3)\cos\theta\cos\ttheta \\
    &&\quad + c_2\Big(2\cos^2\theta\cos^2\ttheta + \sin^2\theta\sin^2\ttheta\Big) \\
    &&\quad + c_3\Big(4\cos^3\theta\cos^3\ttheta + 6\cos\theta\cos\ttheta\sin^2\theta\sin^2\ttheta\Big)\\
    K_1(\theta,\ttheta)
    &=& (c_1-3c_3)\sin\theta\sin\ttheta+4c_2 \cos\theta\cos\ttheta\sin\theta\sin\ttheta\\
    &&\quad +c_3(12\cos^2\theta\cos^2\ttheta\sin\theta\sin\ttheta+3\sin^3\theta\sin^3\ttheta)\\
    K_2(\theta,\ttheta) 
    &=& c_2\sin^2\theta\sin^2\ttheta +6c_3\cos\theta\cos\ttheta\sin^2\theta\sin^2\ttheta\\
    K_3(\theta,\ttheta)
    &=& c_3\sin^3\theta\sin^3\ttheta \\
    K_m(\theta,\ttheta)&=&0\quad (m\ge 4).
    \eeaa
    from the direct computation.

    As in Example \ref{ex:exi}, we consider the case that $c_1>0,c_2=c_3=0$.
    Then, we have \eqref{condi:sta} if $\theta_T$ satisfies
    \beaa
    \cos^2\theta_T < \dfrac{c_1-2c_0}{3c_1}.
    \eeaa 
    Thus, as shown in Fig. \ref{fig:sta}, we can distinguish between parameters with respect to the existence and stability of the spot solution in this case.
\end{example}

\begin{figure}[bt] 
\begin{center}
	\includegraphics[width=14cm]{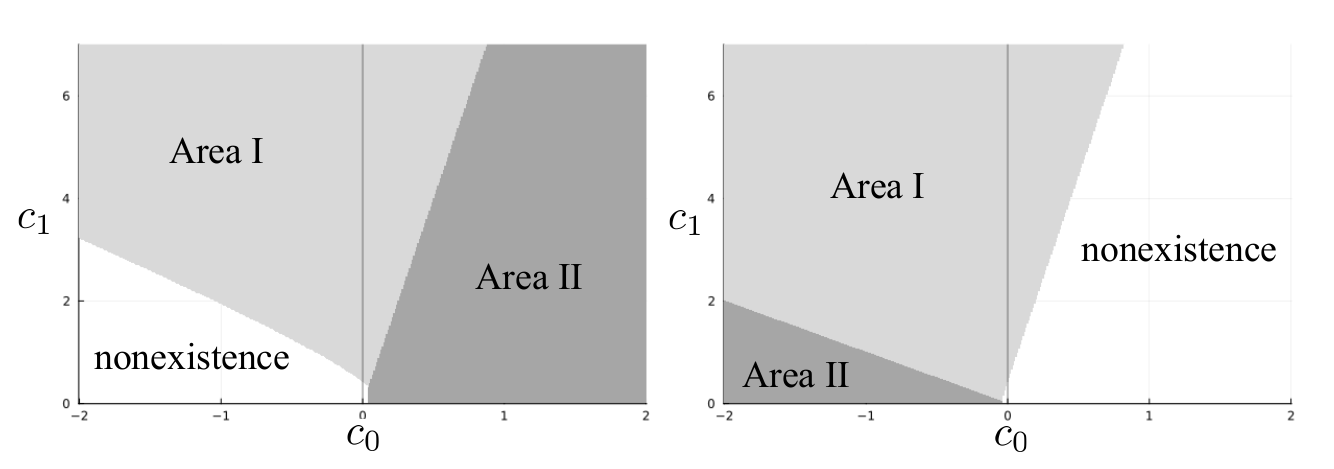}
\end{center}
\caption{\small{
Parameter areas where there exist stationary spot solutions satisfying \eqref{condi:sta} when the kernel is \eqref{kernel} with $c_1>0$ and $c_2=c_3=0$.
In Area I (light gray), there exist a spot solution that is linearized stable.
In Area II (dark gray), there are spot solutions, but they are all unstable.
Left: $u_T=0.5$. Right: $u_T=-0.5$.
}}
\label{fig:sta}
\end{figure}

\medskip
\subsection{Fredholm properties of $L_0$}

We will use perturbation methods to analyze stationary {\red spot} solutions on $S$ and their stability in Section \ref{sec:cri}. 
Then, we need to show that $L_0$ is invertible in some sense in order to apply the implicit function theorem.
Hence, we analyze the properties of $L_0$ here.

Let Assumption \ref{ass:exi} be satisfied.
In this subsection, suppose that
the spot solution \eqref{eq:spot} satisfies \eqref{condi:sta}.
Then, the formal adjoint of $L_0$ is represented by
\beaa
L^{*}_0 w(\bx) &:=& \delta(U(d_0(\bx,N))-u_T) \int_{\bS^2} K(d_0(\bx,\by)) w(\by) {\red ds_0}(\by) - w(\bx) \\
&=& \dfrac{1}{|U'(\theta_T)|}\delta(\theta-\theta_T) \int_{\bS^2} K(d_0(\bx,\by)) w(\by) {\red ds_0}(\by) - w(\bx).
\eeaa
As $L_0 = \cK - I$ is a Fredholm operator from compactness of $\cK$,
we have
\beaa
\mathrm{dim}\mathrm{Ker} L^{*}_0 =  \mathrm{dim}\mathrm{Ker} L_0 =2
\eeaa
from Corollary \ref{cor:ker}.
Moreover, by setting
\be\label{adjoint}
w^1(\bx(\theta,\phi)) := \delta(\theta-\theta_T)\cos\phi,\quad w^2(\bx(\theta,\phi)) := \delta(\theta-\theta_T) \sin\phi,
\ee
we see that $L^{*}_{0}w^{j}=0\ (j=1,2)$.

Based on Fredholm theory for compact operators (see e.g. \cite{Brezis, Kress}), the following result is obtained:
\begin{proposition}\label{prop:fred}
    For $f\in C(\bS^2)$, the equation
    \beaa
    L_0 v(\bx) + f(\bx)=0\quad (\bx\in\bS^2)
    \eeaa
    has a solution $v\in C(\bS^2)$ if and only if $f(\bx)$ satisfies
    \beaa
    \int_{\bS^2} w^1(\bx) f(\bx) ds_0(\bx) 
    =\int_{\bS^2} w^2(\bx) f(\bx) ds_0(\bx) =0.
    \eeaa
\end{proposition}

This proposition implies that $L_0$ is invertible in the suitable subspace of $C(\bS^2)$. 
This allows the implicit function theorem to be applied, and mathematically justifies the construction of the spot solution shown in Theorem \ref{thm:exi} and the characterization of the primary eigenvalue shown in Theorem \ref{thm:sta}.

\bigskip
\section{Stationary spot solutions localized at the north pole on spheroids}\label{sec:cri}

The aim of this section is to derive a stability criterion of stationary spot solutions localized at the north pole on $S$ defined as \eqref{surface}.
We investigate the existence and stability of stationary spot solutions by perturbation methods.
Equation \eqref{eq:main} requires computation for geometric quantities of spheroids such as geodesic distances and Jacobians.
Thus, their perturbation expansions are firstly considered when $\eps$ is sufficiently small.
Next, we construct a stationary spot solution on $S$ based on the result in Section \ref{sec:s2}.
Finally, the primary eigenvalue that determines stability is analyzed.

\medskip
\subsection{Perturbation of the geometric properties}\label{subsec:per}

Let us first consider the geodesic distance.
Geodesics on oblate spheroids have been studied for a long time in connection with geodesy.
For sufficiently small $\eps\in\bR$, $S=S(\eps)$ is geodesically complete from Hopf-Rinow theorem.
Moreover, approximate formulas of geodesic distances for oblate spheroids with sufficiently small flattening have been extensively investigated (e.g. \cite{Andoyer, Karney, Lambert, Thomas}).
Many of these studies are motivated in the context of geodesy and are aimed at deriving accurate geodetic distance formulas for an oblate spheroid with small flattening, such as the Earth.
Among them, Lambert-Andoyer formula is well known as a first-order approximation with respect to flattening \cite{Andoyer, Lambert}.

For $\bx_{\eps}(\theta,\phi), \by_{\eps}(\ttheta,\tphi)\in S$,
we give corresponding points on $\bS^2$ as $\bx(\theta,\phi), \by(\ttheta,\tphi)$, respectively.
Then, Lambert-Andoyer formula is given as
\be\label{exp:dist}
d(\bx_{\eps}(\theta,\phi),\by_{\eps}(\ttheta,\tphi)) =  d_0(\bx(\theta,\phi),\by(\ttheta,\tphi)) - \eps d_p(\bx(\theta,\phi),\by(\ttheta,\tphi)) + O(\eps^2),
\ee
where $d_p:\bS^2\times \bS^2\to [0,+\infty)$ is defined as
\bea
\hspace{1cm}d_p(\bx(\theta,\phi),\by(\ttheta,\tphi))
&:=& \dfrac{1}{2} \left[ 
    (\zeta-\sin\zeta)\left( 
    \dfrac{\cos\left(\dfrac{\theta+\ttheta}{2}\right)\cos\left(\dfrac{\theta-\ttheta}{2}\right)}{\cos\dfrac{\zeta}{2}}\right)^2 \right. \label{distP}\\
 &&   \hspace{2cm} \left. +(\zeta+\sin\zeta)\left( 
    \dfrac{\sin\left(\dfrac{\theta+\ttheta}{2}\right)\sin\left(\dfrac{\theta-\ttheta}{2}\right)}{\sin\dfrac{\zeta}{2}}\right)^2 \right] \notag  
\eea
in which $\zeta=d_0(\bx(\theta,\phi),\by(\ttheta,\tphi))$.
Note that \eqref{exp:dist} holds in the sense of uniform convergence as $\eps\to+0$.
The derivation of the formula is by reference to \cite{Andoyer}, and the higher-order expansion has been studied by \cite{Thomas}.
We proceed with our calculations based on this formula.

Next, we treat the Jacobian.
For $\bx_{\eps}(\theta,\phi)= (1 - \eps p(\eps,\bx(\theta,\phi))) \bx(\theta,\phi)\in S$, we set the corresponding point $\bx=\bx(\theta,\phi)\in \bS^2$.
For simplicity, we write
\beaa
p_{\theta}(\eps,\bx) := \dfrac{\partial}{\partial \theta}[p(\eps,\bx(\theta,\phi))],\quad
p_{\phi}(\eps,\bx) := \dfrac{\partial}{\partial \phi}[p(\eps,\bx(\theta,\phi))].
\eeaa
From the simple computations
\beaa
\dfrac{\partial \bx_{\eps}}{\partial \theta} = \dfrac{\partial \bx}{\partial \theta} 
- \eps \left( p_{\theta} \bx + p  \dfrac{\partial \bx}{\partial \theta} \right),\quad 
\dfrac{\partial \bx_{\eps}}{\partial \phi} = \dfrac{\partial \bx}{\partial \phi} 
- \eps \left( p_{\phi} \bx + p  \dfrac{\partial \bx}{\partial \phi} \right), 
\eeaa
we obtain
\beaa
 \dfrac{\partial \bx_{\eps}}{\partial \theta} \times \dfrac{\partial \bx_{\eps}}{\partial \phi}
= \dfrac{\partial \bx}{\partial \theta} \times \dfrac{\partial \bx}{\partial \phi}
- \eps \left( p_{\phi} \dfrac{\partial \bx}{\partial \theta} \times  \bx +2 p  \dfrac{\partial \bx}{\partial \theta} \times \dfrac{\partial \bx}{\partial \phi} 
+  p_{\theta} \bx\times \dfrac{\partial \bx}{\partial \phi} \right) + O(\eps^2).
\eeaa
As it follows that
\beaa
\left( \dfrac{\partial \bx}{\partial \theta} \times \dfrac{\partial \bx}{\partial \phi} \right) \cdot \left( \dfrac{\partial \bx}{\partial \theta} \times  \bx  \right) = \left( \dfrac{\partial \bx}{\partial \theta} \times \dfrac{\partial \bx}{\partial \phi} \right) \cdot \left(  \bx \times \dfrac{\partial \bx}{\partial \phi}  \right) = 0
\eeaa
from orthogonal relations of the outer product, we have
\beaa
\left| \dfrac{\partial \bx_{\eps}}{\partial \theta} \times \dfrac{\partial \bx_{\eps}}{\partial \phi} \right|
&=& (1 -2 \eps p(\eps,\bx))  \left| \dfrac{\partial \bx}{\partial \theta} \times \dfrac{\partial \bx}{\partial \phi} \right| + O(\eps^2).
\eeaa
By using the fact that
\beaa
\dfrac{\partial \bx}{\partial \theta}(\theta,\phi) \times \dfrac{\partial \bx}{\partial \phi}(\theta,\phi) = (\sin^2\theta \cos\phi, \sin^2\theta \sin\phi, \sin\theta\cos\theta),
\eeaa
we expand Jacobian as
\be\label{exp:jac}
\left| \dfrac{\partial \bx_{\eps}}{\partial \theta}(\theta,\phi) \times \dfrac{\partial \bx_{\eps}}{\partial \phi}(\theta,\phi) \right| = \sin\theta - 2\eps p(\eps,\bx(\theta,\phi)) \sin\theta + O(\eps^2).
\ee

These perturbation terms for the geodesic distance and the Jacobian will play an important role in our perturbation analysis of the spot solution and the stability analysis to be performed later.
\begin{remark}
    The Jacobian perturbation is also computed as \eqref{exp:jac} for the more general $p(\eps,\bx)$.
    However, we are not aware of the perturbation of the geodesic distance in the case of a general perturbed spherical surface.
\end{remark}

\medskip
\subsection{Existence of stationary spot solutions}

Next, let us consider the existence of stationary spot solutions localized at the north pole.
We state the existence result for stationary solutions to Equation \eqref{eq:main}.

Before stating the result, we introduce some notation.
To simplify the analysis, we consider the dependence of various functions on $\theta, \ttheta, \phi, \tphi$.
Let $\tilde{p}(\eps,\theta)=p(\eps,\bx_{\eps}(\theta,\phi))$, where $p(\eps,\bx)$ is defined in \eqref{per}.
Next, the functions $d_0$ and $d_p$, defined as \eqref{dist0} and \eqref{distP}, respectively, are expressed as 
\beaa
d_0(\bx(\theta,\phi),\by(\ttheta,\tphi))=  \tilde{d}_0(\theta,\ttheta,\phi-\tphi)\quad
d_p(\bx(\theta,\phi),\by(\ttheta,\tphi))=  \tilde{d}_p(\theta,\ttheta,\phi-\tphi)
\eeaa
using the appropriate functions $\tilde{d}_0$ and $\tilde{d}_p$.

We shall write the right hand side of Equation \eqref{eq:main} as
\be\label{RHS}
\cL(\eps,\Phi)(\bx_{\eps}) := \int_{S} K(d(\bx_{\eps},\by_{\eps}))H(\Phi(\by_{\eps})-u_T) ds(\by_{\eps}) - \Phi(\bx_{\eps})
\ee
for $\Phi\in C(S)$.
For a stationary spot solution \eqref{eq:spot} to \eqref{eq:main},
we define functions as
\beaa
f(\bx) &:=& \int_{\bS^2}H(U(d_0(\by,N))-u_T)F(\bx,\by)ds_0(\by), \\
F(\bx,\by)&:=&K'(d_0(\bx,\by))d_p(\bx,\by)+2p(\eps,\by)K(d_0(\bx,\by)).
\eeaa
Expressing these in terms of latitude and longitude yields
\beaa
F(\bx(\theta,\phi),\by(\ttheta,\tphi))&=&K'(\tilde{d}_0(\theta,\ttheta,\phi-\tphi))\tilde{d}_p(\theta,\ttheta,\phi-\tphi)+2\tilde{p}(\eps,\ttheta)K(\tilde{d}_0(\theta,\ttheta,\phi-\tphi)) \\
&=:&\tF(\theta,\ttheta,\phi-\tphi),
\eeaa
and 
\beaa
f(\bx(\theta,\phi)) = \int^{\theta_T}_{0} \int^{2\pi}_{0}\tF(\theta,\ttheta,\tphi)\sin\ttheta d\tphi d\ttheta 
=: \tf(\theta).
\eeaa
Here, we note that $\tF(\theta,\ttheta,\tphi)$ is a $2\pi$-periodic function with respect to $\tphi$.

Then, we obtain
\begin{theorem}\label{thm:exi}
    Let Assumption \ref{ass:exi} be satisfied.
    Suppose that the spot solution on $\bS^2$ satisfies Condition \eqref{condi:sta}, which ensures its linearized stability on $\bS^2$.
    Then there are positive constants $\eps_0>0$ and $C_0>0$ such that for any $\eps\in(0,\eps_0)$, 
    there exists a stationary solution $U_{\eps}(\bx_{\eps}(\theta,\phi))= \tilde{U}_{\eps}(\theta)$ to Equation \eqref{eq:main} satisfying
    \beaa
    \sup_{(\theta,\phi)\in [0,\pi]\times[0,2\pi]} |U_{\eps}(\bx_{\eps}(\theta,\phi)) - U(\theta)| \le C_0 \eps.
    \eeaa
    Moreover, we have
    \beaa
    U_{\eps}(\bx_{\eps}(\theta,\phi)) = U(\theta) + \eps V(\theta) + o(\eps),
    \eeaa
    where 
    \bea
    V(\theta) &:=& \frac{2 K_0(\theta,\theta_T) }{\lambda_0 K_1(\theta_T,\theta_T)} \tf(\theta_T)-\tf(\theta), \label{exp:sol} \\
    \lambda_0 &:=& \dfrac{2 K_0(\theta_T,\theta_T)}{K_1(\theta_T,\theta_T)} -1. \notag
    \eea
\end{theorem}
\begin{proof}
    For $\Phi(\bx_{\eps}(\theta,\phi))= U(\theta)+W(\theta)$,
    let us consider the existence of $W$ such that
    \beaa
    G(\eps,W) := \cL(\eps,\Phi)=0
    \eeaa
    for sufficiently small $\eps\in\bR$.
    Set $X:= \{\Phi\in C(\bS^2) \mid \Phi(\bx(\theta,\phi)) = \tilde{\Phi}(\theta)\ (\tilde{\Phi}\in C[0,\pi]) \}$.
    From the symmetry of spheroids, $G:[0,\eps^{*})\times X\to X$ is well-defined and of class $C^1$ for sufficiently small $\eps^{*}>0$.
    Then, we have $G(0,0)=0$.
    Moreover, the Fr\'{e}chet derivatives are given as
    \beaa
    \dfrac{\partial}{\partial \eps}G(0,0) = {\red -\tf},\quad\dfrac{\partial}{\partial W}G(0,0) = L_0.
    \eeaa
    Here, we use \eqref{exp:dist} and \eqref{exp:jac}.
    Since $L_0 :X\to X$ is invertible from Proposition \ref{prop:fred},
    we apply the implicit function theorem.
    Thus, there is $\eps_0>0$ such that 
    there exists $W: [0,\eps_0)\to X$ satisfying
    \beaa
    W(0)=0,\quad G(\eps,W(\eps)) = 0.
    \eeaa
    Moreover, $W(\eps)$ is a $C^1$-function with respect to $\eps$, and 
    $W_1:=\dfrac{\partial W}{\partial \eps}(0)$ is a unique solution of
    \beaa
    L_0 W_1 - f(\bx)=0.
    \eeaa
    From the direct calculation with \eqref{exp} and \eqref{sol:deri}, we find that $V(\theta)$ is a unique solution to this problem.
    Note that $\lambda_0 \neq 0$, since we assume that $U$ satisfies \eqref{condi:sta}.
    Therefore, we obtain the desired assertion.
\end{proof}

\medskip
\subsection{Criterion for the stability on spheroids}

In this subsection, we discuss the stability of the spot solution constructed in Theorem \ref{thm:exi}.
Almost all eigenvalues are negative 
if we consider stable spot solutions on $\bS^2$. 
Thus, only perturbations with 0 eigenvalue should be considered.

Note that Fr\'{e}chet derivative $\cL'(\eps,\Phi)$ of $\cL(\eps)$ at $\Phi\in C(S)$ is represented as
\beaa
\cL'(\eps,\Phi)\Psi(\bx_{\eps}) :=  \int_{S} K(d(\bx_{\eps},\by_{\eps}))\delta (\Phi(\by_{\eps})-u_T) \Psi(\by_{\eps}) ds(\by_{\eps}) - \Psi(\bx_{\eps})
\eeaa
for $\Psi\in C(S)$, where $\cL$ is defined as \eqref{RHS}.
Then, we deduce
\begin{theorem}\label{thm:sta}
    Let all assumptions of Theorem \ref{thm:exi} be satisfied.
    Let $U_{\eps}(\bx_{\eps})$ be a stationary solution to Equation \eqref{eq:main} constructed in Theorem \ref{thm:exi}.
    Then, for sufficiently small $\eps>0$,
    there is an eigenvalue $\mu=\mu(\eps)$ of $\cL'(\eps,U_{\eps})$ such that
    \beaa
    \mu = \eps \mu_1 + o(\eps)
    \eeaa
    holds, where
    \be\label{eig}
    \mu_1 =  \frac{\pi}{|U'(\theta_T)|^2}\partial_\ttheta\left(K_1(\theta_T,\ttheta)V(\ttheta)\sin\ttheta\right)\Bigg|_{\ttheta=\theta_T} -\dfrac{\sin\theta_T}{|U'(\theta_T)|}\int_0^{2\pi}\tF(\theta_T,\theta_T,\tphi)\cos\tphi\ d\tphi.
    \ee
\end{theorem}
\begin{proof}
For simplicity, we identify $\Psi\in C(S)$ and $\Psi\in C(\bS^2)$ by recapturing them at the same latitude and longitude in this proof. 
    Since the eigenspace corresponding to the 0 eigenvalue is two-dimensional, 
    one should focus on the change in each eigenvalue. 
    The method is essentially the same, and thus only one of the cases is presented.

    Let us rewrite $d(\bx(\theta,\phi),\by(\ttheta,\tphi))$ as $\tilde{d}(\theta,\ttheta,\phi-\tphi)$.
    Also, denote $\cL'(\eps,U_{\eps})$ by $\cL'_{\eps}$.
    Note that $\tilde{d}$ is a $2\pi$-periodic even function with respect to $\phi-\tphi$.
    Let $Y := \{ \Psi\in C(\bS^2) \mid \Psi(\bx(\theta,\phi)) = \tilde{\Psi} (\theta)\cos\phi\}$.
    Then, $Y$ is a closed subset of $C(\bS^2)$, and $\cL'_{\eps} Y \subset Y$ holds for any sufficiently small $\eps\in\bR$.
    Moreover, $\mathrm{Im}\ L_0$ is a closed subset of $Y$ from the Fredholm property of $L_0: Y\to Y$.

    For $\Psi\in Y$,
    we introduce the operators into $Y$ as
    \beaa
    \cM_{\eps} \Psi := \cL'(\eps,U_{\eps})\Psi - L_0 \Psi, \quad
    P\Psi := \Psi - \dfrac{\lng \Psi, w^{1} \rng_{L^2(\bS^2)}}{\lng v^{1}, w^{1} \rng_{L^2(\bS^2)}} v^{1},
    \eeaa
    where $v^{1}(\bx)$ and $w^{1}(\bx)$ are defined in \eqref{trans} and \eqref{adjoint}, respectively.
    For $(\mu,\Psi)\in \bR\times Y$,
    let us define $H:\bR^2\times Y \to \mathrm{Im}\ L_0\times\bR^2$ as
    \beaa
    H(\eps,\mu,\Psi) &:=
    & \Bigg( L_0 \Psi - \mu P (v^1 + \Psi)+ P\cM_{\eps} (v^{1}+\Psi), \\
    && \quad\quad \lng \mu (v^1 + \Psi) - \cM_{\eps} (v^{1}+\Psi), w^{1} \rng_{L^2(\bS^2)}, \ 
     \lng v^1, \Psi \rng_{L^2(\bS^2)} \Bigg).
    \eeaa
    For sufficiently small $\eps^{*}>0$, it is easy to see that $H:[0,\eps^{*})\times\bR\times Y \to \mathrm{Im}\ L_0\times \bR^2$ is of class $C^1$.
    We have $H(0,0,0)=(0,0,0)$ and
    \beaa
    DH (\eta,\Theta) &:=&
     \left[ \dfrac{\partial H}{\partial (\mu,\Psi)}(0,0,0)\right] (\eta,\Theta) \\
     &=&
    \Bigg(
        L_0 \Theta - \eta P v^1,
        \eta \lng v^1, w^1 \rng_{L^2(\bS^2)},
        \lng v^1, \Theta \rng_{L^2(\bS^2)} \Bigg)
    \eeaa
    for $(\eta,\Theta)\in \bR\times Y$.
    Notice that $\lng \Psi, w^2 \rng_{L^2(\bS^2)}=0$ holds for any $\Psi\in Y$, 
    we see that $DH:\bR\times Y \to \mathrm{Im}\ L_0\times \bR^2$ is invertible.
    From the implicit function theorem, for sufficiently small $\eps\in\bR$, 
    there are $\mu=\mu(\eps)\in\bR$ and $\Psi(\eps)\in Y$ such that
    $H(\eps,\mu(\eps),\Psi(\eps))=0$.
    Moreover, they satisfy
    \beaa
    \cL'(\eps,U_{\eps})(v^1+\Psi(\eps)) = \mu(\eps) (v^1+\Psi(\eps)).
    \eeaa
    Therefore, $\mu(\eps)$ is an eigenvalue of $\cL'(\eps,U_{\eps})$.

    Next, we consider the perturbation expansion for $\mu(\eps)$.
    Since $\mu(\eps)$ and $\Psi(\eps)$ are of class $C^1$ in the neighborhood of $\eps=0$, 
    they are expandable as 
    \beaa
    \mu(\eps)=\eps \mu_1 +o(\eps),\quad \Psi(\eps) = \eps \Psi_1 + o(\eps).
    \eeaa
    for some $\mu_1\in\bR$ and $\Psi_1\in Y$.
    Let us determine $\mu_1$.
    Focusing on the second component of $H(\eps,\mu(\eps),\Psi(\eps))=0$, 
    the expansion yields 
    \beaa
    \mu_1 \lng v^1, w^{1} \rng_{L^2(\bS^2)} -
    \left< \cN, w^{1} \right>_{L^2(\bS^2)} =0
    \eeaa
    as the coefficient of $O(\eps)$, where
    \beaa
    \cN(\bx):=\dfrac{\partial}{\partial \eps}[\cM_{\eps} v^{1}] \Bigg|_{\eps=0}.
    \eeaa
    Then, by using \eqref{exp:dist}, \eqref{exp:jac} and \eqref{exp:sol}, it is calculated as 
    \beaa
    \cN(\bx(\theta,\phi)) &=& -\frac{1}{|U'(\theta_T)|}\int_0^{2\pi}\partial_\ttheta\left(K(d_0(\bx(\theta,\phi),\by(\ttheta,\tphi)))V(\ttheta)\sin\ttheta\right)|_{\ttheta=\theta_T}\cos\tphi d\tphi \\
    && \quad -\frac{U'(\theta_T)\sin\theta_T}{|U'(\theta_T)|}\int_0^{2\pi}\tF(\theta,\theta_T,\tphi)\cos(\phi-\tphi) d\tphi \\
    &=& -\frac{\pi}{|U'(\theta_T)|}\partial_\ttheta\left(K_1(\theta,\ttheta)V(\ttheta)\sin\ttheta\right)|_{\ttheta=\theta_T}\cos\phi \\
    && \quad -\dfrac{U'(\theta_T)\sin\theta_T}{|U'(\theta_T)|} \left( \int_0^{2\pi}\tF(\theta,\theta_T,\tphi)\cos\tphi d\tphi \right) \cos\phi\\
    \eeaa
    by using \eqref{exp}.
    Moreover, we have
    \beaa
    \left< \cN, w^{1} \right>_{L^2(\bS^2)}
    &=& \sin\theta_T \int^{2\pi}_{0} \cN(\bx(\theta_T,\phi)) \cos \phi d\phi \\
    &=& -\dfrac{\pi^2 \sin\theta_T}{|U'(\theta_T)|}\partial_\ttheta\left(K_1(\theta_T,\ttheta)V(\ttheta)\sin\ttheta\right)|_{\ttheta=\theta_T} \\
    &&\quad -\dfrac{\pi U'(\theta_T)\sin^2\theta_T}{|U'(\theta_T)|} \int_0^{2\pi}\tF(\theta_T,\theta_T,\tphi)\cos\tphi d\tphi \\
    \eeaa
    Note that $\lng v^1, w^{1} \rng_{L^2(\bS^2)}=\pi U'(\theta_T) \sin\theta_T$, then we get \eqref{eig}.

    The above discussion led to the desired assertion. 
    Note that the same eigenvalue is derived even if we proceed with the argument with $v^{2}(\bx)$. 
    This completes the proof.
\end{proof}

\medskip
\begin{remark}
    The integral part of $\cL'(\eps,U_{\eps})$ is a compact operator on $C(S)$, 
    and hence all elements of the spectrum are also eigenvalues except for $-1$, as in $L_0$.
    Furthermore, all eigenvalues of $\cL'(\eps,U_{\eps})$ are represented by a perturbation expansion centered on the eigenvalue of $L_0$, 
    following almost the same argument as in the proof of Proposition \ref{prop:sta}.
\end{remark}

\medskip
Under the assumption of the existence of a stable spot solution on the spherical surface, we construct a stationary spot solution localized at the north pole on oblate spheroids with sufficiently small flattening and derive a stability criterion
\beaa
\mu_1 <0
\eeaa
for it.
For sufficiently small $\eps>0$, the sign of $\mu_1$ determines the stability of this spot solution.
Although the exact form of $\mu_1$ has been calculated, it is not easy to investigate the sign of $\mu_1$ analytically, and this is a future problem.
In the next section, we will examine the value of $\mu_1$ numerically and investigate the global behavior of the spot solution by numerical simulation.

\bigskip
\section{Numerical simulations}\label{sec:num}

Results of numerical simulations are presented. 
The scheme and computation method for solving Equation \eqref{eq:main} numerically are described in Appendix \ref{app:num}.

Hereafter, we set $\eps=0.01$ and use \eqref{kernel} as $K(x)$.
The stability of the spot solution localized at the north pole is shown in Fig. \ref{fig:eig}.
In the figure, for each kernel $K(x)$,
we have determined $u_T$ numerically so that 
\beaa
u_T= \int^{\theta_T}_{0} \left( \int^{2\pi}_{0} K(\cos^{-1} (\sin\theta_T\sin\ttheta\cos\tphi+\cos\theta_T\cos\ttheta)) 
d\tphi \right) \sin \ttheta d\ttheta\ (= U(\theta_T;\theta_T))
\eeaa
for a given $\theta_T\in (0,\pi)$, 
to make it easier to construct the spot solution on $\bS^2$.
We also have confirmed that $U(\theta;\theta_T)$ defined in \eqref{sol} satisfies \eqref{eq:spot} numerically, 
and then checked whether the stability criterion \eqref{condi:sta} on the sphere is satisfied.
In the simulations, the initial conditions are given as exact solutions $U$ in Example \ref{ex:exi}.

In the left panel in Fig. \ref{fig:eig}, 
there is an interval where a stable spot solution exists on the sphere, 
but $\mu_1$ is positive, indicating that the spot solution is unstable on the spheroid. 
On the other hand, the right panel in Fig. \ref{fig:eig} shows that 
there are two intervals with stable spot solutions on the sphere, and that the stability of the corresponding spot solutions on the spheroid is different for each interval.

\begin{figure}[bt] 
\begin{center}
	\includegraphics[width=14cm]{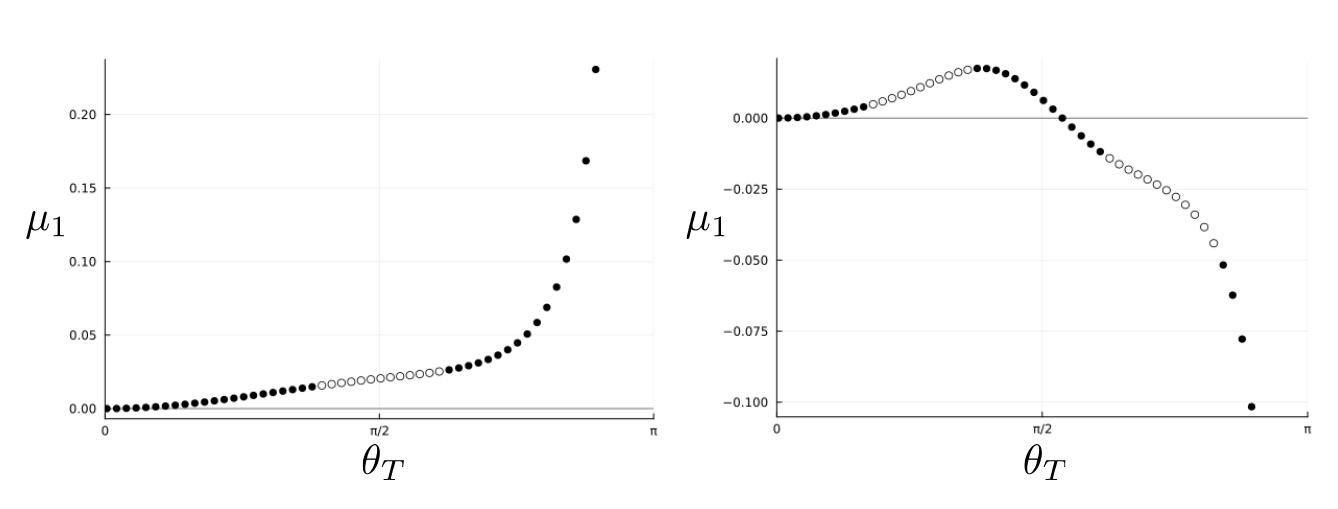}
\end{center}
\caption{\small{
The graph of $\mu_1$. 
The white and black points correspond to the cases when the numerically obtained stationary spot solution $U(\theta;\theta_T)$ on $\bS^2$ satisfies and does not satisfy Condition \eqref{condi:sta}, respectively. 
That is, the white points correspond to the case of a stable spot solution on the sphere.
Left: $c_0=0.6$, $c_1=2.0$, $c_2=c_3=0.0$.
Right: $c_0=0.14$, $c_1=0.9$, $c_2=1.2$, $c_3=0.45$.
The graphs of connection kernels with these parameters are shown in Fig. \ref{fig:kernel}.}}
\label{fig:eig}
\end{figure}

Numerical simulations of Equation \eqref{eq:main} are shown in Fig. \ref{fig:app1}.
Under the same connection kernel as the right panel in Fig. \ref{fig:eig}, 
numerical simulations are conducted to approximate $u_T$ after determining $\theta_T$.
Moreover, the simulations correspond to the case when there are stable spot solutions on the sphere and $\mu_1$ is positive and negative, respectively.
When $\mu_1$ is positive, the spot solution moves away from the north pole and is localized near the equator.
On the other hand, when $\mu_1$ is negative, the spot solution moves toward the north pole and is localized there.
Our analysis has not been able to construct spot solutions on the spheroid that are localized at other points than the poles, 
but numerical simulations suggest the existence of spot solutions localized at the equator.

Finally, Fig. \ref{fig:app2} shows the different dynamics by different connection kernels on the same spheroid and for a fixed $u_T$.
In this simulation, after determining model parameters such as $K$ and $u_T$,
we numerically obtained $\theta_T\in(0,\pi)$ that can construct a stable stationary spot solution.
This example suggests that even if the curved surface is determined, 
the dynamics will vary depending on the shape of the connection kernel.

\begin{figure}[bt] 
\begin{center}
	\includegraphics[width=15cm]{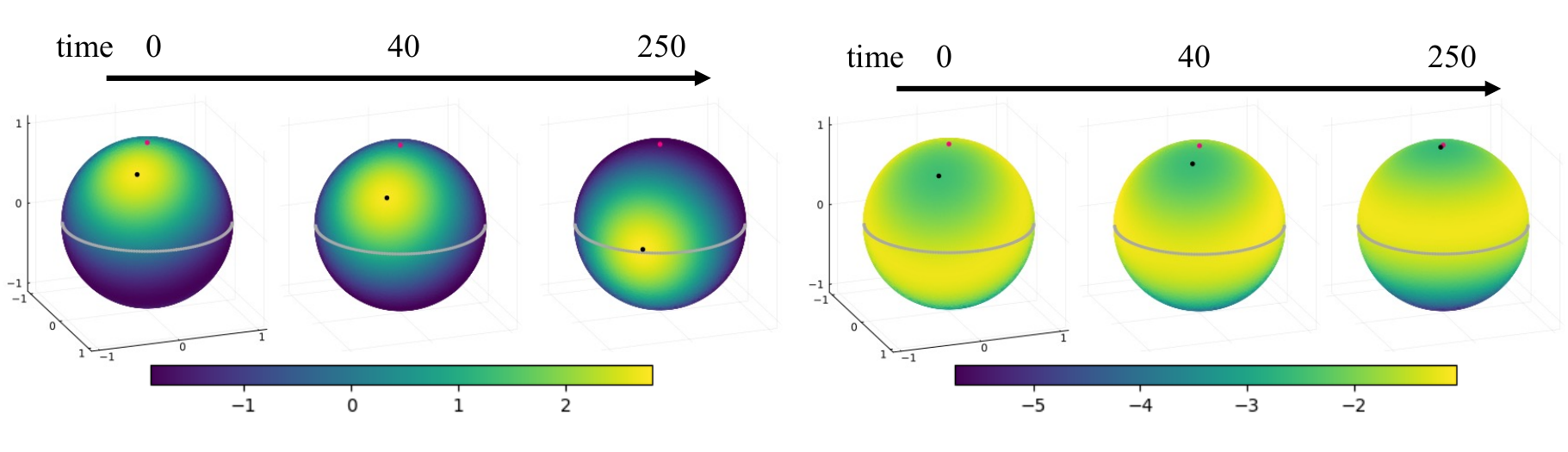}
\end{center}
\caption{\small{
Numerical simulations of Equation \eqref{eq:main} with $\eps=0.01$.
The red dot and the gray line represent the north pole and the equator, respectively.
The black dot corresponds to the center of the spot solution obtained numerically.
Parameters of the kernel function \eqref{kernel} are given as $c_0=0.14$, $c_1=0.9$, $c_2=1.2$, $c_3=0.45$ in both panels. 
Left: $\theta_T=1.0$, which corresponds to $\mu_1>0$ and $u_T\simeq -0.1898$.
Right: $\theta_T=2.0$, which corresponds to $\mu_1<0$ and $u_T\simeq -2.6068$.
}
}
\label{fig:app1}
\end{figure}

\begin{figure}[bt] 
\begin{center}
	\includegraphics[width=14cm]{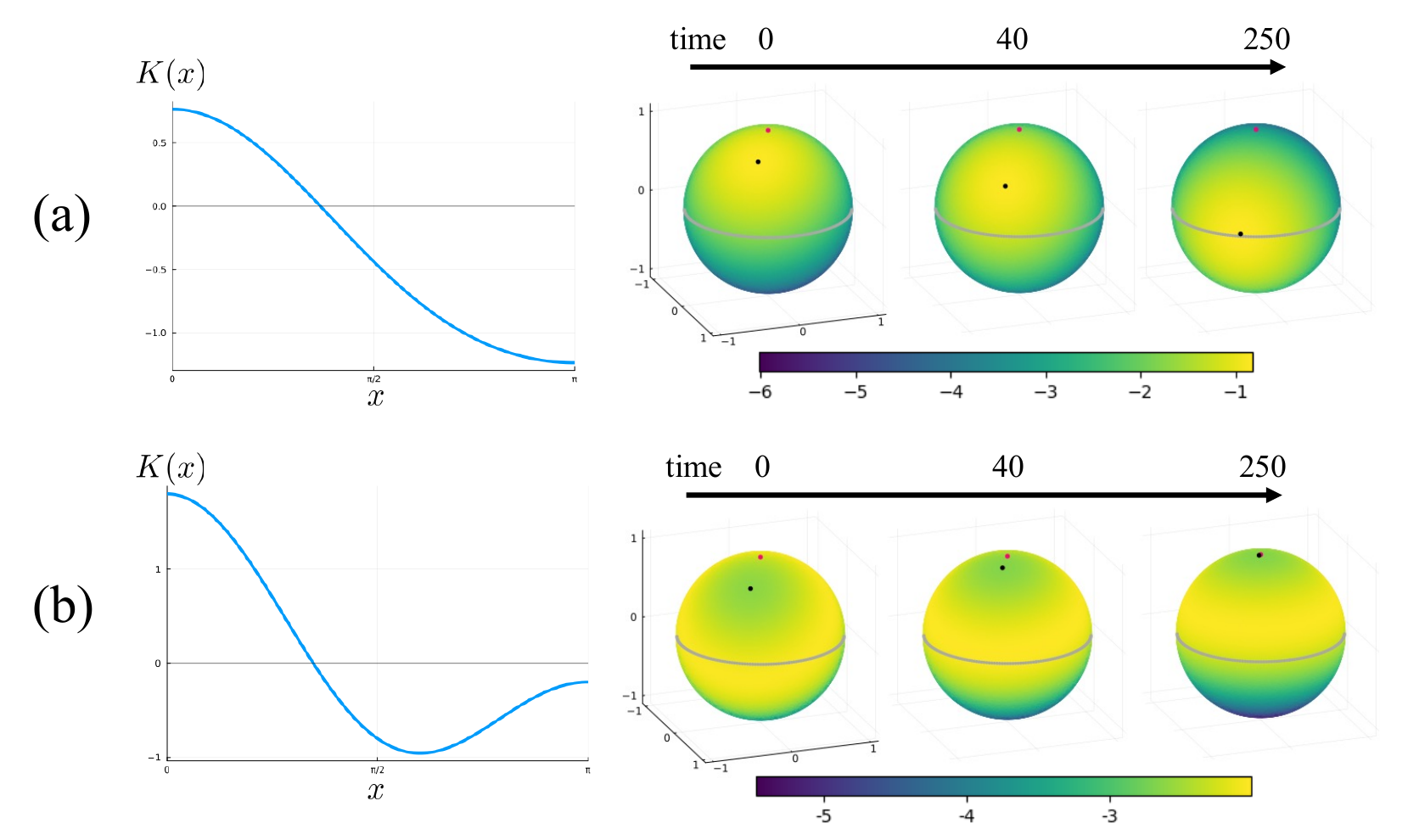}
\end{center}
\caption{\small{
Numerical simulations of Equation \eqref{eq:main} with $\eps=0.01$ and $u_T=-4.32$.
The left and right panels show the graphs of the kernel and the numerical results, respectively.
(a) $c_0=-0.335$, $c_1=1.0$, $c_2=0.1$, $c_3=0.0$, which corresponds to $\theta_T\simeq 1.9989$ and $\mu_1>0$.
(b) $c_0=0.0$, $c_1=1.0$, $c_2=0.8$, $c_3=0.0$, 
which corresponding to $\theta_T\simeq 2.3963$ and $\mu_1<0$.
}
}
\label{fig:app2}
\end{figure}

\bigskip
\section{Concluding remarks}\label{sec:dis}

We have discussed the existence and stability of spot solutions on spherical surfaces and spheroids, which are examples of curved surfaces.
In the case of spherical surfaces, due to rotational symmetry, our results can be used to construct spot solutions localized at arbitrary points under appropriate variable transformations. 
On the other hand, this is not the case for spheroids, and our results are applicable only to spot solutions localized at the north and south poles. 
In this case, we should discuss where the spot solution can be localized with respect to the latitudinal direction since it is axisymmetric. 
Numerical simulations suggest that the spot solution can be localized at the equator. 
Our analysis also does not exclude the possibility of spot solutions localized at locations other than the poles or equator. 
It is interesting to determine where the spot solution can be stably localized given the connection kernel and the curved surface, but this is an open problem.

It has also been shown that the derived stability criterion on the spheroid is numerically verifiable, and it has been discussed that the dynamics of the spot solution around the poles can be understood.
However, it is not known how the stability criterion depends on the features of the curved surface.
As noted in \cite{MCC}, the dynamics of the spot solution may depend on the Gaussian curvature, but this study was unable to make the connection.
It is exciting to be able to make specific characterizations, but this is an open question.

The derivation of equations for the motion of spot solutions is also an important topic to consider.
Many theories have been reported to derive differential equations of the position of localized patterns in Euclidean space \cite{Bress2, Ei, EI}.
Although it may be possible to derive the equation of motion by following their procedure, it remains to be considered how to describe the time evolution of the position of the spot solution on the surface and what kind of information can be extracted.
This is a future problem to be solved.

We also obtained various analytical results by assuming the mean firing rate function to be the Heaviside function.
More general choices, such as sigmoidal firing rate functions, were not considered here, and it is worth investigating how the pattern behavior changes with different nonlinearities.
Additionally, while our analysis focused on the scalar Amari model, a similar analysis could be extended to two-variable neural field systems incorporating a recovery variable, as in \cite{MCC}.
Understanding how this additional variable, in conjunction with surface geometry, affects the dynamics is an interesting open question.

The applicability of our analysis to more general curved surfaces is also mentioned.
In the case of perturbed spherical surfaces, 
most of the discussion can proceed in a similar manner.
The analysis of spot solutions and their primary eigenvalues on the perturbed spherical surface require information on the perturbation terms of the geodesic distance and the Jacobian.
Although the Jacobian perturbation is the same as in Section \ref{subsec:per},
the perturbation of the geodesic distance is generally not easy to analyze.
Perturbation analysis of geodetic distances can extend the adaptive range of this analysis, allowing analysis of complicated perturbed spherical spheres.
On the other hand, our analysis is not directly applicable to curved surfaces with large deformations of spherical surfaces, and different analysis methods need to be considered.

Finally, we discuss the effects of surface geometry on pattern formation in the neural field equations.
Several studies have used eigenfunction expansions of the solution to explain pattern formation \cite{PAO, VNFC}.
In the case of the damped wave model treated in \cite{PAO}, the eigenvalues and eigenfunctions of the Laplace-Beltrami operator depend on the surface geometry, 
and therefore, the various dynamics can be explained by the analysis of the given surface.
On the other hand, 
our analysis of a neural field equation on a spheroid suggests that the same curved surface can produce different dynamics depending on the shape of the connection kernel.
This means that pattern formation is not entirely determined by only determining the curved surface in our setting.
Therefore, to analyze pattern formation on a curved surface, it is important to properly understand not only the geometry of a curved surface but also its relationship to the spatial interaction.
As mentioned in Section \ref{sec:int}, 
various spatio-temporal patterns have been reported in neural field equations. 
Further theoretical studies of those patterns on curved surfaces will lead to a better understanding of the geometric constraints of the excitation patterns.

\bigskip
\section*{Acknowledgements}
HI is supported by JSPS KAKENHI Grant Numbers 23K13013 and 24H00188.

\bigskip
\section*{CRediT authorship contribution statement}
{\bf Hiroshi Ishii}: Conceptualization, Formal analysis, Investigation, Methodology, Writing-original draft, Writing–review and editing.
{\bf Riku Watanabe}: Formal analysis, Software, Visualization, Writing–review and editing.

\bigskip
\section*{Declaration of Competing Interest}
The authors declare that they have no known competing financial interests or personal relationships that could have appeared to influence the work reported in this paper.

\bigskip
\section*{Declaration of generative AI and AI-assisted technologies in the writing process}
During the preparation of this work the authors used DeepL and Grammarly in order to improve their English writing. 
After using this tool/service, the authors reviewed and edited the content as needed and take full responsibility for the content of the publication.

\bigskip
\appendix

\section{Computation of Example \ref{ex:exi}}\label{app:ex}

Here, we introduce the computation for the case $c_1>0,\ c_2=c_3=0$ in Example \ref{ex:exi}.
By setting $C(x):=-c_1x^3+(c_1-2c_0)x$, we have
\beaa
U(\theta_{T};\theta_{T}) = \pi C(\cos\theta_T) + 2\pi c_0.
\eeaa
To obtain existence conditions, 
it is sufficient to find the maximum and the minimum of $C(x)$ on $(-1,1)$.
Since we have $C'(x)=-3c_1x^2+c_1-2c_0$, we consider the following three cases:
\beaa
\mathrm{(i)}\ c_1\le 2c_0,\quad \mathrm{(ii)}\ c_1>2c_0,\ c_0+c_1>0,\quad \mathrm{(iii)}\ c_0+c_1\le 0.
\eeaa

(i) We obtain $C'\le 0$ on $[-1,1]$ and thus conclude that 
\beaa
\min_{x\in [-1,1]} C(x) =C(1) = -2c_0,\quad 
\max_{x\in [-1,1]} C(x) =C(-1) = 2c_0.
\eeaa

(ii) Since we have $\sqrt{\dfrac{c_1-2c_0}{3c_1}}\in (0,1)$ and $C'\left(\pm\sqrt{\dfrac{c_1-2c_0}{3c_1}}\right)=0$,
we get
\beaa
\min_{x\in [-1,1]} C(x) &=&\min\left\{C(1),C\left(-\sqrt{\dfrac{c_1-2c_0}{3c_1}}\right)\right\}, \\ 
\max_{x\in [-1,1]} C(x) &=&\max\left\{C(-1),C\left(\sqrt{\dfrac{c_1-2c_0}{3c_1}}\right)\right\}.
\eeaa
Here, we note that 
\beaa
C\left(\sqrt{\frac{c_1-2c_0}{3c_1}}\right)=\frac{2}{3}\sqrt{\frac{c_1-2c_0}{3c_1}}(c_1-2c_0)>0.
\eeaa
The following inequality is useful to evaluate them:
\beaa
(y + 1)^2(1 - 8y)=(1 - 2y)^3 - 27y^2
\begin{cases} >0 & \left(-1<y<\dfrac{1}{8}\right), \vspace{3mm} \\
<0 & \left(y>\dfrac{1}{8}\right).
\end{cases}
\eeaa
By using this inequality, 
we obtain
\beaa
C(-1) \le C\left(\sqrt{\dfrac{c_1-2c_0}{3c_1}}\right) 
\eeaa
when $c_1\ge 8c_0$.
Introducing additional two conditions
\beaa
\mathrm{(iia)}\ c_1 < 8c_0,\quad
\mathrm{(iib)}\ c_1 \ge 8c_0,
\eeaa
we obtain the following:
\beaa
\mathrm{(iia)}&\Rightarrow&\ \min_{x\in [-1,1]} C(x) =C\left(1\right),\quad 
\max_{x\in [-1,1]} C(x) ={\red C\left(-1\right)}, \\
\mathrm{(iib)}&\Rightarrow&\
 \min_{x\in [-1,1]} C(x) =C\left(-\sqrt{\dfrac{c_1-2c_0}{3c_1}}\right),\quad 
\max_{x\in [-1,1]} C(x) =C\left(\sqrt{\dfrac{c_1-2c_0}{3c_1}}\right) \\.
\eeaa

(iii) We deduce $C'\ge 0$ on $[-1,1]$ and hence obtain 
\beaa
\min_{x\in [-1,1]} C(x) =C(-1) = 2c_0,\quad 
\max_{x\in [-1,1]} C(x) =C(1) = -2c_0.
\eeaa

Notice that if we take care to exclude the case $x=\pm 1$, the condition for the existence of $\theta_T\in(0,\pi)$ is the assertion in Example \ref{ex:sta}.

\section{Numerical method}\label{app:num}

\renewcommand{\theequation}{B.\arabic{equation}}

This section explains the method for the numerical simulations in Section \ref{sec:num}.
To simulate Equation \eqref{eq:main} numerically,
we construct the code by the Julia programming language with some Python packages via \texttt{PyCall}. 
The code is based on the open code in \cite{NI1}.
Our code for numerical simulations is available as follows:
\begin{center}
\url{https://github.com/RikuWatanabe-git/neural-field-model-on-spheroid.git}
\end{center}

\medskip
Let us first explain how to construct a triangle mesh. 
We set $J$ points on $\bS^2$ by subdividing an icosahedron.
For a given $\eps\in\bR$, let $\bx_j=\bx(\theta_j,\phi_j)\in S\ (j=1,...,J)$ be the point generated by the projection of the $j$-th node in $\bS^2$.
The triangle mesh is constructed by \texttt{scipy.spatial.Delaunay} and \texttt{scipy.spatial.ConvexHull} in \texttt{SciPy} library in Python.
Let $\cS$ be the triangulation of $S$ with vertices $\{\bx_{j}\}_{j=1,\ldots,J}$.
For a point $\bx_{j}$, let $\Lambda_j$ be the set of all pairs of points $(\bx_{j_1},\bx_{j_2})$ such that $j_1<j_2$
and the triangle $[\bx_j,\bx_{j_1},\bx_{j_2}]$ is in $\cS$.

Next, we introduce the numerical scheme.
To approximate the integral term, 
we apply the formulation of area elements on the triangle mesh in \cite{Xu}.
We give the following quantity $s_j\ (j=1,...,J)$ as an area element corresponding to a point {\red $\bx_j$}:
\beaa
s_j:=\sum_{(\bx_{j_1}\bx_{j_2})\in \Lambda_j}\frac{1}{8}\left(\frac{|\bx_{j_1}-\bx_j|^2}{ \tan\zeta_1}+\frac{|\bx_{j_2}-\bx_j|^2}{\tan\zeta_2}\right).
\eeaa
Here we set 
\beaa
\zeta_1=\arccos\left(  \frac{\bx_j-\bx_{j_1}}{|\bx_{j}-\bx_{j_1}|} \cdot \frac{\bx_{j_2}-\bx_{j_1}}{|\bx_{j_2}-\bx_{j_1}|}\right),\quad 
\zeta_2=\arccos\left( \frac{\bx_j-\bx_{j_2}}{|\bx_{j}-\bx_{j_2}|} \cdot  \frac{\bx_{j_1}-\bx_{j_2}}{|\bx_{j_1}-\bx_{j_2}|}\right).
\eeaa
For a time step $\delta_t>0$, 
consider $u^{j}_n \simeq u(n\delta_t,\bx_j)\ $ and
approximate Equation \eqref{eq:main} as
\beaa
u^{n+1}_i=u_i^n+\delta_t\left(\sum\limits_{\{j|u_j^n\geq u_T\}}K(d_{i,j})s_j-u_i^n\right),
\eeaa
where $d_{i,j}:=d(\bx_i,\bx_j)\ (1\le i,j\le J)$.
We calculate each $d_{i,j}$ by \texttt{geographiclib.geodesic} library.

Let us discuss the validity of the numerical scheme.
We evaluate the validity of the scheme in the case of the spherical surface, i.e., $\eps=0$.
First, the number of meshes is evaluated by the absolute error between the spot solution $U(\theta)$ discussed in Section \ref{sec:s2} and the numerical solution calculated.
We use \eqref{kernel} as the connection kernel with the same parameters in the right panel of Fig. \ref{fig:kernel}.
Also, we set the model parameter 
\beaa
u_T \simeq -0.1898
\eeaa
to have a stable spot solution with $\theta_T=0.1$, as shown in the left panel of Fig. \ref{fig:app1}.

For a given number of meshes $J\in\bN$, 
the numerical solution $u_{num}(t,\bx)$ was calculated by giving $u_{num}(0,\bx_j)= U(\theta_j)\ (j=1,2,\ldots,J)$ and the absolute error was calculated as 
\be\label{abs-err}
\max_{1\le j \le J} |u_{num}(50.0,\bx_{j})-U(\theta_j)|.
\ee
The left panel of Fig. \ref{fig:B1} shows the absolute error when the numerical solutions are computed with $\delta_t = 10^{-3}$.
It can be seen that the error decreases monotonically with the number of meshes $J$.
Based on this result, we have set the number of meshes to $J = 40962$ and performed the numerical calculations in the paper.

For the time step $\delta_t$, we set the initial state to 
\beaa
u_{num}(0,\bx_{j}) = e^{-10.0\times \theta^2_j}
\eeaa
and evaluated the absolute error between the numerical solutions as we varied $\delta_{t}$. 
The purpose of this is to investigate the stability of the numerical scheme and the range of $\delta_t$ within which the results do not change significantly.
Let us denote the numerical solution corresponding to $\delta_t$ by $u_{num}(t,\bx;\delta_t)$.
In the right panel of Fig. \ref{fig:B1}, 
we calculated 
\be\label{abs-err2}
\max_{1\le j \le J} |u_{num}(t,\bx_{j};10^{-4})-u_{num}(t,\bx_j;\delta_t)|.
\ee
to compare with the numerical solution.
Since the numerical computation is completed faster when $\delta_t$ is large, we want to choose $\delta_t$ that is larger than $\delta_t = 10^{-4}$ and has a smaller error.
Therefore, we chose $\delta_t = 10^{-3}$ to perform the numerical simulations in the paper.

\begin{figure}[bt] 
\begin{center}
	\includegraphics[width=14cm]{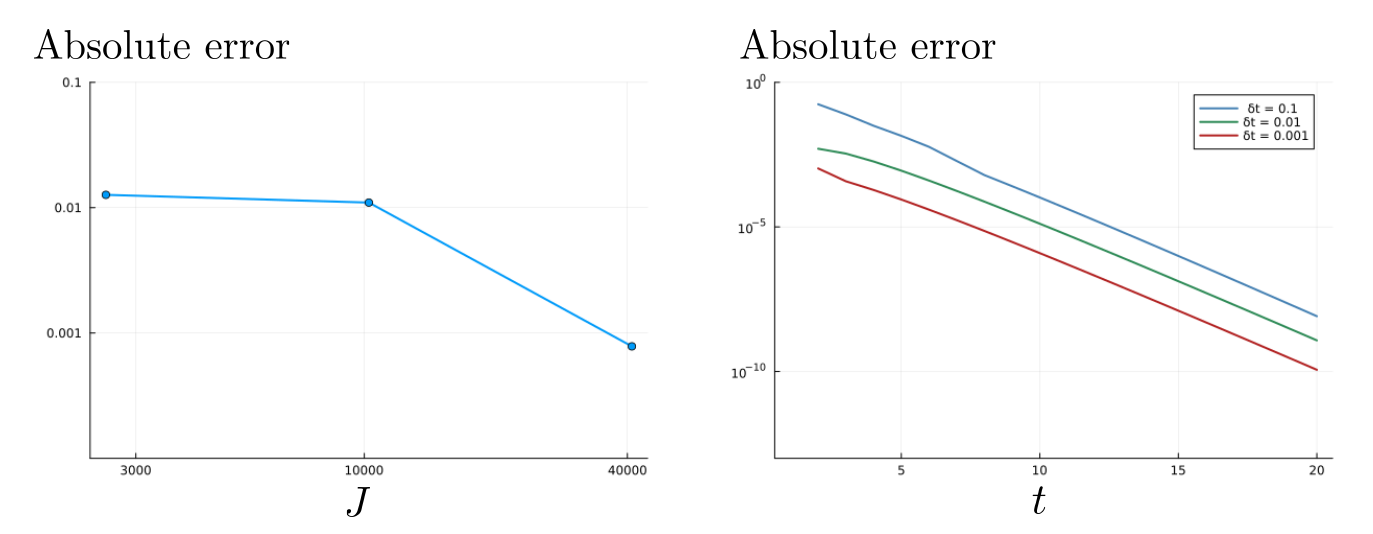}
\end{center}
\caption{\small{
Left: Absolute error depending on the number of meshes. The horizontal axis is the number of meshes $J$ and the vertical axis is the absolute error \eqref{abs-err}.
Right: Absolute error depending on the time step.
The horizontal axis is $t$ and the vertical axis is the absolute error \eqref{abs-err2}.
Blue: $\delta_t=10^{-1}$. Green: $\delta_t = 10^{-2}$. Red: $\delta_t=10^{-3}$.
}}
\label{fig:B1}
\end{figure}


\end{document}